\documentclass[12pt]{article}
\usepackage{amsmath,amsthm}
\usepackage{amssymb,latexsym}
\usepackage{enumerate}
\usepackage[T1]{fontenc}

\usepackage[dvips]{psfrag}
\usepackage{amsmath}
\usepackage{amsthm}
\usepackage{amssymb}
\usepackage{amscd}
\usepackage{amsfonts}
\usepackage{amsbsy}
\usepackage{graphicx,color}         
\newtheorem{theorem}{Theorem}
\newtheorem{lemma}{Lemma}

\newtheorem{example}{Example}
\newtheorem{remark}{Remark}
\DeclareMathOperator{\sech}{sech}
\begin{document}
\baselineskip=17pt

\title{Taylor Series for Adomian Decomposition Method}
\author{Ekaterina Kutafina\\
Department of Applied Mathematics\\
 AGH University of Science and Technology\\
al. Mickiewicza 30\\
30-059 Krak\'ow, Poland\\
E-mail: kutafina@agh.edu.pl}
\date{}

\maketitle

\begin{abstract} 
In this paper we analyse the exact solutions to scalar PDEs obtained thanks to summable Taylor series provided by Adomian's decomposition method. We propose a modification of the method which makes the calculations of Taylor coefficients easier and more direct. The difference is essential for instance in case of nonhomogenous equations or initial conditons and is illustrated by some examples.
\end{abstract}
{\bf Keywords:} Adomian decomposition method; exact solutions; Taylor series\\
 \emph{  2000 Mathematics Subject Classification}:	35C10, 35C05,	65D15.

\section{Introduction}
The area of exact solutions to nonlinear differential equations has become 
very popular in recent decades, when the development of personal computers
enabled more efficient work with known algorithms. Adomian decomposition method
\cite{adomian1,adomian2,wwbook} in the matter of fact was developed to find approximated solutions to differential equations, but in many publications \cite{wwbook,wwfisher} we can find interesting examples where obtained power series were actually summable to exact solutions.
 Typical way to obtain such solutions is to sum up certain Taylor series. In our paper we are going to present some situations when it seems reasonable to use modified techniques to obtain Taylor series. In mathematical physics we often have to deal with scalar PDEs of
space and time variables $x,t\in \mathbb{R}$. We will show that for nonautonomous
equations of this type the method could be easily modified to get Taylor series directly. In order to explain our idea let us first briefly present the classical method, so it would be easier to show the differences. 

\section{Adomian decomposition method}
Let us consider the following  one-dimensional equation in Cauchy-Kovalevska form:
\begin{equation}\label{eqn}u_t(x,t)=F(u,u_x,u_{xx},...,u_{x^n})+g(x),\end{equation}
 where  $x,t\in \mathbb{R}$, $u_t=\frac{\partial u}{\partial t}$, $u_{x^i}=\frac{\partial^i u}{\partial {x^i}}$,  $F(p_0,p_1,p_2,p_3,...,p_{n})$ is a polynomial function of it's arguments and $g(x)$ is an analytic non-autonomous term.  The  corresponding initial condition is $u(x,0)=f_0(x)$. This quite special case would be perfectly sufficient to present the advantages of proposed modification.\\
Let us introduce the auxiliary notation: $G[u]=G(u,u_x,u_{xx},...,u_{x^n})$ for any functional  $G$ defined on jet-space. 
In classical approach LHS of (\ref{eqn}) is usually split into two parts: $F[u]=L_F[u]+N_F[u]$,  where $L_F[u]$ is a linear operator with respect to $u,u_x,...,u_{x^n}$ while $N_F[u]$ is nonlinear part of $F[u]$. Then the operator
\[L^{-1}(.)=\int_0^t(.)\,dt \] can be introduced to express the solution of (\ref{eqn}) in the form:
\[u=f_0(x)+g(x)\,t+\int_0^t L_F[u]+N_F[u]\, dt.\] Next we assume $u=\sum_{i=0}^\infty u_i$ and consequently $L_F[u]=\sum_{i=0}^\infty L_F[u_i]$ and $N_F[u]=N_F[\sum_{i=0}^\infty]=\sum_{i=0}^\infty A_i$. The newly introduced terms $A_i$ are so-called Adomian's polynomials, which could be obtained e.g. with the help of following formulae:
\[A_i=\frac{1}{n!}\frac{d^n}{d\lambda_n}[F(\sum_{i=0}^n \lambda^i u_i)].\] In the paper \cite{wwpoly} author present very intuitive way to obtain these polynomials.
The idea could be easily understood from the example below. 
\begin{example}\label{uux} For instance if we take nonlinearity in the form $ N_F[u]=u\,u_x$ then \[N_F[u]=(u_0+\epsilon u_1+\epsilon^2 u_2+...)(u_{0x}+\epsilon u_{1x}+\epsilon^2u_{2x}+...)=\]
\[=u_0\,u_{0x}+\epsilon (u_1\,u_{0x}+u_0\,u_{1x})+...\]
and here $A_i$ would be a coefficient at $\epsilon^i$. \end{example} Let us notice that "$\epsilon$"-notation was not used in original paper, but in our opinion it makes the choice of $A_i$ more clear. Let us underline that each polynomial $A_i$ is dependent only on the functions $u_0,...,u_i$. No higher orders are involved. Going back to the decomposition algorithm:
\[u=f_0(x)+g(x)\,t+\int_0^t \sum_{i=0}^\infty L_F[u_i]+\sum_{i=0}^\infty A_i \,dt,\]
therefore the following recurrence could be defined:
\[u_0=f_0(x)+g(x\,)t\]
\[u_1=\int_0^t L_F[u_{0}]+A_{0}\, dt\]
\[...\]
\[u_n=\int_0^t L_F[u_{n-1}]+A_{n-1}\, dt\]
\[...\]
\[u=\lim_{k\to \infty}\sum_{i=0}^k u_i.\]
\begin{remark}\label{rem1} If $g(x)=0$ then there exists a sequence of functions $h_i(x)$ such as $u_n=h_n(x)\,t^n$.\end{remark}
\begin{proof} By induction $u_0=f_0(x)=:h_0(x)$. Let us assume, that $u_i=h_i(x)\,t^i$ so $t$ appears in the same power as $\epsilon$ in the example \ref{uux} which means that $A_i=k_i(x)\,t^i$ and
\[u_{i+1}=\int_0^t L_F[h_{i}(x)\,t^i]+k_i(x)t^i dt=\frac{t^{i+1}}{i+1}(L_F[h_i(x)]+k_i(x))=:h_{i+1}t^{i+1}.\]\end{proof}
\begin{remark}\label{rem2} If $f_0(x)=0$ then there exists a sequence of functions $h_i(x)$ such as $u_n=h_n(x)\,t^{n+1}$.\end{remark}
\begin{proof}By induction $u_0=t\,g(x)=:h_0(x)$. Let us assume, that $u_i=h_i(x)\,t^{i+1}$ which means that $A_i=k_i(x)\,t^{i+1}$ and
\[u_{i+1}=\int_0^t L_F[h_{i}(x)\,t^{i+1}]+k_i(x)t^{i+1} dt=\frac{t^{i+2}}{i+2}(L_F[h_i(x)]+k_i(x))=:h_{i+1}t^{i+2}.\]\end{proof}
These remarks imply that in case of autonomous equation or zero initial condition the algorithm leads straight to Taylor series.
\section{Main Results}
In our following research it would be comfortable to skip dividing $F[u]$ into two parts. The whole functional $F[u]$ could be as well approximated by Adomian polynomials. Let us also denote by $u^{(k)}=\sum_{i=0}^{k} u_i$.  The key difference would be the fact that now we choose polynomials in different way. Let $B_0=F[u_0]$, but \[B_i=F[u^{(i)}]-F[u^{(i-1)}]\] for $i=1,2,...$ .
Comparing to the example \ref{uux} now we obtain $B_0=A_0=u_0\,u_{0x}$, but for a change $B_1=F[u_0+u_1]-F[u_0]=u_0\,u_{1x}+u_1\,u_{0x}+u_1\,u_{1x}$. However it is obvious that still $\lim_{i \to \infty}B_i=F[u]$. The new recurrence is:
\[u_0=f_0(x)+g(x\,)t\]
\[u_1=\int_0^t F[u_0]dt=\int_0^t B_0 \,dt\]
\[u_2=\int_0^t B_1\,dt=\int_0^t F[u_0+u_1]-F[u_0]\,dt\]
\[...\]
\[u_k=\int_0^t F[u^{(k-1)}]-F[u^{(k-2)}]\,dt.\]
\begin{theorem}\label{main} Let us consider a partial differential equation in the form $u_t(x,t)=F[u(x,t)]+g(x)$, together with the initial condition $u(x,0)=f_0(x)$ where $x,t \in \mathbb{R}$, $u: \mathbb{R}^2\to\mathbb{R}$, $F[u(x,t)]=F[u(x,t),u_x(x,t),...,u_{x^n}(x,t)]$, $u_{x^i}=\frac{\partial^i u}{\partial x^i}$. We also assume, that $F[u(x,t)]$ is analytical of it's arguments and $F[0]=0$. Then formal Taylor series for the solution $u(x,t)$ could be found using formula (\ref{form}).
\end{theorem}
\begin{proof}

Let us start with the summation:
\[u^{(k)}=\sum_{i=0}^{k} u_i=u_0+\int_0^t F[u^{(k-1)}]\,dt.\]
Therefore using Taylor's formula
\[u^{(k+1)}=u_0+\int_0^t F[u^{(k)}]\,dt=\]
 \[=u_0+\int_0^t F[u^{(k)}]|_{t=0}+\left[ \frac{\partial F[u^{(k)}]}{\partial u^{(k)}}\frac{\partial u^{(k)}}{\partial t}+ \frac{\partial F[u^{(k)}]}{\partial u_x^{(k)}}\frac{\partial u_x^{(k)}}{\partial t}+...+\frac{\partial F[u^{(k)}]}{\partial u_{x^n}^{(k)}}\frac{\partial u_{x^n}^{(k)}}{\partial t}\right]_{t=0}\,t+\]
{\footnotesize\[+\left[\frac{\partial^2 F[u^{(k)}]}{\partial( u^{(k)})^2}\left(\frac{\partial u^{(k)}}{\partial t}\right)^2+...+2 \frac{\partial^2 F[u^{(k)}]}{\partial u^{(k)}\partial u_x^{(k)}}\frac{\partial u^{(k)}}{\partial t}\frac{\partial u_x^{(k)}}{\partial t}+...+\frac{\partial F[u^{(k)}]}{\partial u_{x^n}^{(k)}}\frac{\partial^2 u_{x^n}^{(k)}}{\partial t^2}\right]_{t=0}\frac{t^2}{2!}+....\,dt=\]
\[=u_0+ F[u^{(k)}]_{t=0}\,t+\left[\frac{\partial F[u^{(k)}]}{\partial u^{(k)}}\frac{\partial u^{(k)}}{\partial t}+ \frac{\partial F[u^{(k)}]}{\partial u_x^{(k)}}\frac{\partial u_x^{(k)}}{\partial t}+...+\frac{\partial F[u^{(k)}]}{\partial u_{x^n}^{(k)}}\frac{\partial u_{x^n}^{(k)}}{\partial t}\right]_{t=0}\frac{t^2}{2!}+\]
\[+\left[\frac{\partial^2 F[u^{(k)}]}{\partial( u^{(k)})^2}\left(\frac{\partial u^{(k)}}{\partial t}\right)^2+...+2 \frac{\partial^2 F[u^{(k)}]}{\partial u^{(k)} \partial u_x^{(k)}}\frac{\partial u^{(k)}}{\partial t}\frac{\partial u_x^{(k)}}{\partial t}+...+\frac{\partial F[u^{(k)}]}{\partial u_{x^n}^{(k)}}\frac{\partial^2 u_{x^n}^{(k)}}{\partial t^2}\right]_{t=0}\,\frac{t^3}{3!}+....=\]}

Before we continue let us formulate the following lemma:

\begin{lemma}\label{lemma} If $u^{(i)}=a_0+a_1\,t+a_2\,t^2+...$ and  $u^{(i+1)}=b_0+b_1\,t+b_2\,t^2+...$ then $a_s=b_s$ for $s\leq i$.\end{lemma}
\begin{proof}The statement holds if and only if $u_{i+1}=t^{i+1}\,h_{i+1}(x,t)$ for some analytical function $h_k$ and $k\geq 1$.
\\  Basis: for $k=1$ $u_1=\int_0^t F[u_0]\,dt$. Since $F$ is analytical, it could be written in series form:
\[F(u,u_x,...,u_{x^n})=\sum_{i_0,i_1,...,i_n}^{}b_{i_0\,i_1\,...i_n}u^{i_0}
u_x^{i_1}...u_{x^n}^{i_n}.\] 
Thus $F$ also could be written as series with respect to $t$, $F[u_0]=c_0+c_1\,t+...$ and after integration $u_1=c_0\,t+\frac{c_1}{2!}\,t^2+...$ .\\ 
 Inductive step: we assume, that $u_k=t^kh_k(x,t)$, then 
\[u_{k+1}=\int_0^t F[u_0+u_1+...+u_k]-F[u_0+u_1+...+u_{k-1}]\,dt=\int_0^t F[S+t^kh_k(x,t)]-F[S]\,dt,\]
where $S=u_0+u_1+...+u_{k-1}$. Using the series form:
{\footnotesize \[u_{k+1}=\int_0^t \sum_{i_0,i_1,...,i_n}b_{i_0\,i_1\,...i_n}((S+t^kh_k(x,t))^{i_0}
(S+t^kh_k(x,t))_x^{i_1}...(S+t^kh_k(x,t))_{x^n}^{i_n}-S^{i_0}
S_x^{i_1}...S_{x^n}^{i_n})\,dt\]}
In the main theorem we assumed $F[0]=0$ so $b_{00...0}=0$. Thus the smallest possible power of $t$ in the sum is $t^k$ and after integration the proof is completed.
 \end{proof}
 Going back to the main proof:
{\footnotesize \[=u_0+ F[f_0]\,t+\left[\frac{\partial F}{\partial u}[f_0](g(x)+F[f_0])+ \frac{\partial F}{\partial u_x}[f_0](g(x)+F[f_0])_x+...+\frac{\partial F}{\partial u_{x^n}}[f_0](g(x)+F[f_0])_{x^n}\right]\frac{t^2}{2!}+\]
\[+ \Big[\frac{\partial^2 F}{\partial u^2}[f_0]\left(g(x)+F[f_0]\right)^2+...+2 \frac{\partial^2 F}{\partial u \, \partial u_x}[f_0](g(x)+F[f_0])(g(x)+F[f_0])_x+... \]
\[ +\frac{\partial F}{\partial u_{x^n}}[f_0]\left( \frac{\partial F}{\partial u}[f_0](g(x)+F[f_0])+ \frac{\partial F}{\partial u_x}[f_0](g(x)+F[f_0])_x+...+\frac{\partial F}{\partial u_{x^n}}[f_0](g(x)+F[f_0])_{x^n}\right)\Big]\,\frac{t^3}{3!}+...\,.\]}
In other words, after taking limit we obtain the following formal form:

\begin{equation}\label{form}u(x,t) \approx a_0+a_1\,\frac{t}{1!}+a_2\,\frac{t^2}{2!}+...,  \end{equation}
where
\[a_0=f_0\]
\[a_1=g(x)+F[a_0]\]
\[a_2=\frac{\partial F}{\partial u}[a_0]\,a_1+\frac{\partial F}{\partial u_x}[a_0]\,a_{1x}+...+\frac{\partial F}{\partial u_{x^n}}[a_0]\,a_{1x^n}\]
\[a_3=\frac{\partial^2 F}{\partial u^2}[a_0]a_1^2+\frac{\partial^2 F}{\partial u_x^2}[a_0]a_{1x}^2+...+2\frac{\partial^2 F}{\partial u\,\partial u_x}[a_0]a_1a_{1x}+...+\frac{\partial F}{\partial u}[f_0]a_{2}+...+\frac{\partial F}{\partial u_{x^n}}[f_0]a_{2x^n}\]
\[...\]
Further terms could be easily obtained using formulae  for higher differentials and the key fact, that $u^{k}$ and $u^{k+1}$ have the same coefficients up to $k$th power.
\end{proof}

\section{Examples}
\begin{example} Let us start with the example from \cite{wwfisher}. Authors considered the following Fisher's equation
\begin{equation}\label{fish} u_t=u_{xx}+6\,u(1-u)\end{equation} together with the initial condition \[u(x,0)=\frac{1}{(1+e^x)^2},\]
and obtain the exact solution of (\ref{fish}) using Adomian decomposition method. Here $F[u]=F(u, u_{xx})$, $\frac{\partial F(u, u_{xx})}{\partial u}=6-12\,u$,
 $\frac{\partial F(u, u_{xx})}{\partial u_{xx}}=1$, $\frac{\partial^2 F(u,u_{xx})} {\partial u^2}=1$, $f_0=\frac{1}{(1+e^x)^2}$,
\[a_0=f_0\]
\[a_1=F[f_0]=\left(\frac{1}{(1+e^x)^2}\right)''+6\,\frac{1}{(1+e^x)^2}\left(1-\frac{1}{(1+e^x)^2}\right)=
\frac{10 e^x}{(1+e^x)^3}\]
\[a_2=(6-12 f_0)a_1+a_{1xx}=50\frac{e^x(2e^x-1)}{(e^x+1)^4}\]
\[a_3=-12(a_{1})^2+(6-12 f_0)a_2+a_{2xx}=250\frac{4e^{2x}-7e^x+1}{(e^x+1)^5}\]
\[etc..\]
The result after summation repeats the cited paper: \[u(x,t)=\frac{1}{(1+e^{x-5t})^2}.\].
\end{example}
As we have seen from propositions \ref{rem1}, \ref{rem2} the classical method does not lead straight to Taylor's series only if the equation or initial condition is nonhomogenous. The next two examples cover both situations.
\begin{example}
Now let us consider the non-autonomous heat equation \cite{wwbook}: \[u_t=u_{xx}+\sin x,\] with the initial condition $u(0,x)=\cos x$.
Using classical approach we obtain:
\[u_0=\cos x+t\sin x\]
\[u_1=-t\cos x-\frac{1}{2!}t^2\sin x\]
\[u_2=\frac{1}{2!}t^2\cos x+\frac{1}{3!}t^3\sin x\]
\[....\]
Applying theorem \ref{main} we can directly obtain Taylor series (the only non-zero derivative is $\frac{\partial F}{\partial u_{xx}}$):
\[a_0=\cos x\]
\[a_1=\sin x+[\cos x]_{xx}=\sin x -\cos x\]
\[a_2=\frac{\partial F}{\partial u_{xx}}a_{1xx}=-\sin x+\cos x\]
\[a_3=\frac{\partial F}{\partial u_{xx}}a_{2xx}=\sin x-\cos x\]
\[etc.,\]
so finally
\[u(x,t)\approx \cos x+(\sin x-\cos x)(t-\frac{t^2}{2!}+\frac{t^3}{3!}-...)=\cos x e^{-t}+\sin x(1-e^{-t}).\]\end{example}
To complete the illustration we choose the example with nonlinearity and nonautonomous term.
\begin{example} The following inhomogeneous advection problem is solved in \cite{wwbook}:
\[u_t+u\,u_x=x,\qquad u(x,0)=2.\]
With the help of decomposition method the following recursive relations was obtained:
\[u_0=2+x\,t\]
\[u_1=-t^2-\frac{1}{3}xt^3\]
\[u_2=\frac{5}{12}t^4+\frac{2}{15}xt^5\]
\[...\]
Meanwhile using  theorem \ref{main}
\[f_0=2,\quad g(x)=x,\quad F_u[u]=-u_x,\quad F_u[2]=0,\]\[\quad F_{u_x}[u]=-u\quad  \quad F_{u_x}[2]=-2,\quad F_{u\,u_x}=-1\] and all other derivatives vanish in $f_0$.
\[a_0=2\]
\[a_1=x\]
\[a_2=-2\]
\[a_3=-2x\]
\[etc..\]
In both methods we obtain
\[2\left(1-\frac{1}{2!}t^2+\frac{5}{4!}t^4+... \right)+x\left(t-\frac{1}{3}t^3+\frac{2}{15}t^5\right)=\]
\[=2\, \sech\, t+x\, \tanh\, t.\]
\end{example}

\section*{Acknowledgements} 
The presented  research was partially supported
by the Polish Ministry of Science and Higher Education.

\end{document}